\documentclass{article}

\usepackage{amsmath}
\usepackage{amsfonts}
\usepackage{amssymb}
\usepackage{amsthm}
\usepackage{enumerate}
\usepackage{bbm}
\usepackage{hyperref}

\newtheorem{theorem}{Theorem}[section]
\newtheorem{lemma}[theorem]{Lemma}

\newtheorem{definition}[theorem]{Definition}

\newtheorem{fact}[theorem]{Fact}

\newcommand{\suchthat}{\;\ifnum\currentgrouptype=16 \middle\fi|\;}

\newcommand{\F}{\mathbb{F}}
\newcommand{\R}{\mathbb{R}}

\newcommand{\E}{{\rm I\kern-.3em E}}

\begin{document}
\title{A short note on the joint entropy of $n/2$-wise independence}
\author{
Amey Bhangale\footnote{Department of Computer Science, Rutgers University. Research supported in part by NSF grants CCF-1253886 and CCF-1540634.}  \thanks{{\tt amey.bhangale@rutgers.edu} }\and
Aditya Potukuchi\footnotemark[1] \thanks{ {\tt aditya.potukuchi@cs.rutgers.edu}.}
}

\maketitle

\begin{abstract} 
In this note, we prove a tight lower bound on the joint entropy of $n$ unbiased Bernoulli random variables which are $n/2$-wise independent. 

For general $k$-wise independence, we give new lower bounds by adapting Navon and Samorodnitsky's Fourier proof of the `LP bound' on error correcting codes.

This counts as partial progress on a problem asked by Gavinsky and Pudl{\'a}k in \cite{GP}.
\end{abstract}


\section{Introduction}

In this note, we study the Shannon entropy of unbiased Bernoulli random variables that are $k$-wise independent. The Shannon entropy (or simply, entropy) of a discrete random variable $X$, taking values in a set $Y$, is given by $\operatorname{H}(X) = -\sum_{y \in Y}\Pr(X = y)\log(\Pr(X = y))$, where all logarithms are base $2$. A joint distribution on $n$ unbiased, Bernoulli random variables $X = (X_1, \ldots, X_n)$ is said to be $k$-wise independent if for any set $S \subset [n]$ with $|S| \leq k$, and any string $a \in \{0,1\}^k$, we have that $\Pr(X|_S = a) = \frac{1}{2^{|S|}}$, where $X|_S$ means $X$ restricted to the coordinates in $S$.

Bounded independence distributions spaces come up very naturally in the study of error correcting codes. Let $\mathcal{C}$ be a binary linear code over $\F_2$ of dimension $k$, distance $d$, and length $n$, i.e., $\mathcal{C}$ is (also) a linear subspace of $\F_2^n$ of dimension $k$. Let $M$ be the $(n-k) \times n$ \emph{parity check matrix} for $\mathcal{C}$ (i.e., $\mathcal{C} = \operatorname{nullspace}(M)$). It can be checked that every $d-1$ columns of $M$ are linearly independent. So, the random variable $y^TM$, where $y$ is uniformly distributed in $\F_2^{n-k}$, s $(d-1)$-wise independent. This connection can be used to construct $k$-wise independent sample spaces of small support. For $k = O(1)$, \emph{BCH codes} give $k$-wise independent sample spaces of support size $O(n^{\frac{k}{2}})$. And for $k = n/2$, using the \emph{Hadamard code}, one gets a sample space of support size $\leq \lceil \frac{2^n}{n+1} \rceil$. It can be shown that these sample spaces are optimal in support size.

The study of entropy of joint distributions with bounded dependence was first studied by Babai in~\cite{Bab}. In~\cite{GP}, Gavinsky and Pudl{\'a}k prove asymptotically tight lower bounds on the joint entropy of $k$-wise independent (not necessarily Bernoulli) random variables for small values of $k$. They prove that such a distribution must have entropy at least $\log {n \choose k/2}$. This implies the previously stated lower bound on the size of the support, as it is more general (since $\operatorname{H}(X) \leq \log |\operatorname{supp}(X)|$). Here, we study the case when $k = \Theta(n)$ and in particular, we show asymptotically tight bounds when $k = n/2 - o(n)$.  We state the results.

\begin{theorem}
\label{thm:main}
Let $X$ be a joint distribution on unbiased Bernoulli random variables $(X_1, X_2, \ldots, X_n)$ which is $k-1$-wise independent, then 
\[
\operatorname{H}(X) \geq n - n\operatorname{H}\left(\frac{1}{2} - \sqrt{\frac{k}{n}\left(1-\frac{k}{n}\right)} \right) - o(n).
\]
\end{theorem}

Here, for a number $p \in (0,1)$, we say $\operatorname{H}(p)$ to mean $- p \log p - (1-p) \log (1-p)$, i.e., the entropy of a $p$-biased Bernoulli random variable. The case where $k = n/2$ is especially simple, and conveys most of the main idea, so we prove it separately in Section~\ref{sec2}. 

\begin{theorem}
\label{thm:half}
Let $X$ be a joint distribution on unbiased Bernoulli random variables $(X_1, X_2, \ldots, X_n)$ which is $n/2$-wise independent, then $\operatorname{H}(X) \geq n - \log(n+1)$.
\end{theorem}

Our proof follows Navon and Samorodnitsky's \cite{NS} approach to the the \emph{Linear Programming bound} for error correcting codes (also known as the \emph{MRRW} bound, \cite{MRRW}). This approach uses Fourier analysis and a covering argument. Our main observation is that these techniques essentially prove a lower bound on the \emph{Renyi entropy} of any $k$-wise independent distribution, which then gives us a lower bound for the (Shannon) entropy. 

\section{Preliminaries}

The (basically spectral) argument is stated in the language of Fourier analysis, as in~\cite{NS}. Henceforth, for a random variable $Y = Y(x)$, we say $\E_x[Y(x)]$ (or simply $\E[Y]$) to mean the expected value of $Y$ when $x$ is drawn uniformly from $\{0,1\}^n$. For a function $f:\{0,1\}^n \rightarrow \R$, the \emph{Fourier
decomposition} of $f$ is given by 
$$f(x) = \sum_{S \subseteq [n]}
\widehat f(S) \chi_S(x) ,$$
$\text{ where } \chi_S(x) := 
(-1)^{\mathop{\sum}_{i\in S} x_i} \text{ and }\widehat f(S) :=
\E_{x}[ f(x)\chi_S(x)].$

For any two functions $f,g :\{0,1\} \rightarrow \R$, we also have an inner product, given by

\[
\langle f, g \rangle = \E_x[f(x)g(x)]
\]

\begin{theorem}[Plancherel's identity]
For any $f,g : \{0,1\}^n \rightarrow \R$,
$$\langle f, g\rangle = \sum_{S\subseteq [n]} \widehat{f}(S)\cdot\widehat{g}(S).$$
\end{theorem}
For $f,g : \{0,1\}^n \rightarrow \R$, the convolution of $f$ and $g$ denoted by $f*g$ is defined as:
$$ (f*g)(x) = \E_{y}[f(y)g(y+x)].$$

\begin{fact}
Let $f,g : \{0,1\}^n \rightarrow \R$, then $\widehat{(f*g)}(S) = \widehat{f}(S)\cdot\widehat{g}(S)$ for all $S\subseteq [n]$.
\end{fact}

Next, we define the \emph{R\'{e}nyi entropy}.

\begin{definition}[R\'{e}nyi Entropy]
For a random variable $X$ supported on a finite set $Y$, the \emph{R\'{e}nyi Entropy}, denoted by $\operatorname{H}_2(X)$ is given by:

\[
\operatorname{H}_2(X) = -\log \left(\sum_{y \in Y}p(y)^2 \right),
\]

where $p(y) = \Pr(X = y)$.
\end{definition}

The following is a well known relation between entropy and the R\'{e}nyi entropy:

\begin{fact}
\label{lem:entropy}
For a random variable $X$ of finite support size, $\operatorname{H}(X) \geq \operatorname{H}_2(X)$
\end{fact}

\begin{proof}
Let $p_1, \ldots, p_t$ be the (nonzero) probabilities on the support of $X$. Since $\log$ is a concave function, from Jensen's Inequality, we have 

\[
\log \left(\sum_{i \in [t]}p_i^2 \right) \geq \sum p_i \log(p_i),
\]

which proves this fact.
\end{proof}

For this proof, we will also look at the hypercube $\{0,1\}^n$ as a graph.

\begin{definition}[Hamming graph]
The hamming graph $H_n = (V_n, E_n)$ is a graph with vertex set $V_n = \{0,1\}^n$, and edges $\{x,y\} \in E_n$ if $x$ and $y$ differ on exactly one coordinate.
\end{definition}

\section{Entropy of $n/2$-wise independent distributions}
\label{sec2}

Here, we give the proof of Theorem~\ref{thm:half}. For a random variable $X$, we define a function $f (= f_{X}) : \{0,1\}^n \rightarrow \mathbb{R}_{\geq 0}$ to be the \emph{normalized probability density function}, i.e., 

\[
f(x) = 2^n \cdot \Pr(X = x).
\]

So, we have $\E[f] = 1$.

\begin{proof}[Proof of Theorem~\ref{thm:half}]
Let $f:\{0,1\}^n \rightarrow \R_{\geq 0}$ be the normalized probability density function of an $n/2$-wise independent distribution of Bernoulli random variables $X$. Let $A$ denote the adjacency matrix of the Hamming graph $H_n$.\\

Let $L: \{0,1\}^n \rightarrow \R$ such that $L(x) = 1$ iff $|x| = 1$ and $0$ otherwise. First, we observe that for any function $f$, $Af = L \ast f$. Also $\widehat{L}(S) = n - 2|S|$.  We have,
\begin{align*}
\langle Af, f \rangle & = \langle L \ast f, f \rangle \\
& = \sum_{S \subseteq [n]}(\widehat{L \ast f})(S) \cdot \widehat{f}(S) \tag*{(By Plancherel's identity)}\\
& = \sum_{S \subseteq [n]}\widehat{L}(S) \cdot \widehat{f}(S)^2 \\
& = \widehat{L}(\emptyset)\widehat{f}(\emptyset)^2 + \sum_{\substack{S \subseteq [n], \\ 1\leq |S|\leq n/2}}\widehat{L}(S) \cdot \widehat{f}(S)^2 +\sum_{\substack{S \subseteq [n],\\  |S|> n/2}}\widehat{L}(S) \cdot \widehat{f}(S)^2.
\end{align*}
We now use that fact that $f$ is a normalized pdf of $n/2$-wise independent distribution and hence $\widehat{f}(S)=0$ for all $1\leq |S|\leq n/2$. Thus, we can upper bound $\langle Af, f \rangle$ as
\begin{align*}
\langle Af, f \rangle &= n\widehat{f}(\emptyset)^2 + 0 +  \sum_{\substack{S \subseteq [n],\\  |S|> n/2}}\widehat{L}(S) \cdot \widehat{f}(S)^2  \\
& \leq  n\widehat{f}(\emptyset)^2 -  \sum_{\substack{S \subseteq [n],\\  |S|> n/2}} \widehat{f}(S)^2  \\
 & = n \widehat{f}(\emptyset)^2 + 1 - \sum_{S \subseteq [n]}\widehat{f}(S)^2 \\
 & = n + 1 - \E[f^2].
\end{align*}
Since $\langle Af, f \rangle \geq 0$, we have $\E[f^2] \leq n + 1$. Let $p_1, p_2, \ldots, p_t$ be the set of nonzero probabilities on the support of the distribution. We have that $\operatorname{H}_2(X) = -\log (\sum p_i^2) \leq n - \log(n+1)$. By Fact~\ref{lem:entropy}, we have $\operatorname{H}(X) \geq n - \log(n+1)$.

\end{proof}

\emph{Remark:} The bound obtained above is tight when $n + 1$ is a power of $2$. In the usual way, we identify $\{0,1\}^n$ with $\F_2^n$. The tight case is constructed from the \emph{Hadamard code}. Let $P$ be the parity check matrix of the Hadamard code, so $Pv = 0$ for codewords $v$. It can be checked that the uniform distribution on the row space of $P$ is $n/2$-wise independent. Since this a uniform distribution on $\frac{2^n}{n+1}$ points, we have the required bound.

\section{Entropy of $k$-wise independent distributions where $k = \Theta(n)$}

We carry over the notation from the previous section. 
For a subset $B\subseteq \{0,1\}^n$, define $\lambda_B$ as
$$ \lambda_{B} = \max \left\{\frac{\langle Af, f\rangle}{\langle f, f \rangle} \bigg| f: \{0,1\}^n \rightarrow \R, \mathtt{supp}(f)\subseteq B \right\}.$$\\

For general $k-1$-wise independent balanced Bernoulli distributions where $k = \Theta(n)$, we have the following approach: The main idea is that given a $k-1$-wise independent distribution $X$ given by the density function $f$, we make another random variable $Z$, given by density function $g$ as follows: sample a point, according to $f$, and perturb it randomly to some nearby point according to some distribution. Then use an argument similar to the previous section on this new distribution. Formally, let $Y$ be a random variable that is supported on the hamming ball of radius $r$ with center $0^n$. We have a new random variable $Z = X \oplus Y$. 

There are three useful facts about this distribution on $Z$, as follows:

\begin{fact}
\label{fact:ballr}

\begin{enumerate}
\item[(a)] The resulting distribution $Z$ is also $(k-1)$-wise independent.
\item[(b)] $\operatorname{H}(X) + \operatorname{H}(Y) \geq \operatorname{H}(Z)$.
\item[(c)] $g = f\ast d$.
\end{enumerate}
where $d: \{0,1\}^n \rightarrow \R$ is normalized density function of $Y$ supported on the hamming ball of radius $r$ around the origin.
\end{fact}
\begin{proof}
The proof of \emph{(a)} is that since $\hat{g}(S) = \hat{f}(S)\cdot \hat{d}(S)$, it implies that $\hat{g}(S) = 0$ for all $0<|S|\leq (k-1)$.

Item \emph{(b)} is true because $X$ and $Y$ collectively determine $Z$.

To prove \emph{(c)}, let $\mathtt{wt}(x)$ denotes the hamming weight of $x \in \{0,1\}^n$. By definition, we have
\begin{align*}
(f\ast d)(x) &= \frac{1}{2^n}\sum_{y\in \{0,1\}^n} f(y) d(x+y)\\
& = 2^n \sum_{y : \mathtt{wt}(x+y)\leq r} \Pr(X = y) \cdot \Pr(Y = x+y) \\
& = 2^n \Pr(Z=x)\\
& = g(x)
\end{align*}

\end{proof}

Next, we make use of the following lemma from~\cite{NS} to obtain bounds on the maximum eigenvalue of the (Hamming) graph induced on a Hamming ball:

\begin{lemma}
\label{lem:eigenvalue}
Let $B_r$ be a Hamming ball of radius $r$, then we have:
\[
\lambda_{B_r} \geq 2\sqrt{r(n-r)} - o(n)
\]
\end{lemma}

We omit the proof of the above lemma since we are going to use it exactly as is presented in~\cite{NS}. Now we can choose the distribution $d$ as the normalized eigenfunction of the hamming ball, i.e., the function for which:

\begin{equation}
\label{eq:def_d}
\frac{\langle Ad, d\rangle}{\langle d,d \rangle} = \lambda_{B_r}
\end{equation}
Further, we have that $d$ is a nonnegative function, with $\E[d] = 1$, and $Ad \geq \lambda_{B_r} d$, and $d$ is only supported on the Hamming ball of radius $r$. Denote $\lambda_r = \lambda_{B_r}$ for convenience.

Now, we are ready to give the proof of Theorem~\ref{thm:main}
\begin{proof}[Proof of Theorem~\ref{thm:main}]
Let $f$ be the normalized probability density function of a $k-1$-wise independent distribution. Let $g = f \ast d$ where $d$ satisfies Equation~\ref{eq:def_d}. The thing to note is that for $S \subseteq [n]$, since $\widehat{g}(S) = \widehat{f}(S)\widehat{d}(S)$, we have that $\widehat{g}(S) = 0$ for $0<|S| < k$. Again, we look at the quantity $\langle Ag, g \rangle$:

\begin{align}
\langle Ag, g \rangle & = \langle L \ast g, g \rangle \nonumber\\
& = \sum_{S \subseteq [n]}(\widehat{L \ast g})(S) \cdot \widehat{g}(S) \nonumber\\
& = \sum_{S \subseteq [n]}\widehat{L}(S) \widehat{g}^2(S) \nonumber\\
& \leq n \cdot\widehat{g}^2(\emptyset) + 1 + (n - 2k)\sum_{S \subseteq [n]}\widehat{g}^2(S) \nonumber\\
& = n + (n-2k)\E[g^2]. \label{eq:ub}
\end{align}
On the other hand, we have:
\begin{align}
\langle Ag, g \rangle & = \langle L \ast (d \ast f) , g \rangle \nonumber\\
& = \langle (L \ast d) \ast f , g \rangle \nonumber\\
	& \geq \langle (\lambda_r d) \ast f , g \rangle \nonumber\\
    & = \lambda_r \langle d \ast f , g \rangle \nonumber\\
	& = \lambda_r \langle g,g \rangle \nonumber\\
	& = \lambda_r \E[g^2]. \label{eq:lb}
\end{align}
Combining (\ref{eq:ub}) and (\ref{eq:lb}), we have,
\[
(\lambda_r - (n - 2k))\E[g^2] \leq n.
\]

We choose $r$ such that $\lambda_r \geq n - 2k + 1$, this gives us the upper bound $\E[g^2] \leq n$. Using Fact~\ref{lem:entropy} and the definition of $g$, we get
\begin{align*}
\operatorname{H}(Z) & \geq \operatorname{H}_2(Z)\\
& = - \log\left(\frac{\E[g^2]}{2^n}\right)\\
& \geq - \log\left(\frac{n}{2^n}\right) = n - \log n.
\end{align*}

By Fact~\ref{fact:ballr} (2), $\operatorname{H}[Y] + \operatorname{H}[X] \geq \operatorname{H}[Z] \geq n - \log n$, giving us $\operatorname{H}[X] \geq n - \log n - \operatorname{H}[Y]$.  Since $Y$ is supported on the hamming ball of radius $r$, we just use the trivial bound $\operatorname{H}(Y) \leq \log\left({n \choose r}\right)$. Now, using the well-known upper bound $\sum_{i=0}^{r}{n \choose i} \leq 2^{nH(r/n)}$, we get $H(Y) \leq  n \operatorname{H}\left(\frac{r}{n} \right)$ and hence

$$\operatorname{H}(X) \geq n-n\operatorname{H}\left(\frac{r}{n} \right) - \log n.$$

The value $r$ for our purpose is $\frac{n}{2} - \sqrt{k(n-k)} + o(n)$ satisfying $\lambda_r \geq n - 2k + 1$ which, by Lemma~\ref{lem:eigenvalue}, completes the proof.
\end{proof}

Since the best known size lower bound goes by proving a lower bound on the $\ell^2$ norm, it easily extends to entropy, which, by Jensen's Inequality, is shown to be a `weaker' quantity.

\section{Acknowledgements}

We would like to thank Swastik Kopparty for the many very helpful discussions, and suggestions for the writeup.



\begin{thebibliography}{9}
\bibitem{AS}
N. Alon and J. H. Spencer, \emph{The Probabilistic Method}, Third Edition, Wiley, 2010.
\bibitem{Bab}
L. Babai,
\emph{Entropy Versus Pairwise Independence} (Preliminary version)
{http://people.cs.uchicago.edu/$\sim$laci/papers/13augEntropy.pdf}, 2013.
\bibitem{GP}
D. Gavinsky and P. Pudl{\'a}k,
\emph{On the joint entropy of $ d $-wise-independent variables},
Commentationes Mathematicae Universitatis Carolinae, 57, 3,   pages 333--343,2016
\bibitem{MRRW}
R. McEliece, E. Rodemich, H. Rumsey and L. Welch,
\emph{New upper bounds on the rate of a code via the Delsarte-MacWilliams inequalities},
IEEE Transactions on Information Theory, vol. 23, no. 2, pp. 157-166, Mar 1977.
\bibitem{NS}
M. Navon, A. Samorodnitsky, 
\emph{Linear programming bounds for codes via a covering argument},
Discrete and Computational Geometry, 41, 2, 2009.
\end{thebibliography}


\end{document}